\documentclass[nofootinbib,twocolumn,showpacs,superscriptaddress,twoside,pra]{revtex4-1}
\usepackage{amsmath}
\usepackage{amssymb}
\usepackage{amsfonts}
\usepackage{enumerate}
\usepackage{mathrsfs}
\usepackage{graphicx}
\usepackage{epstopdf}
\usepackage{enumerate}
\usepackage{changepage}

\usepackage{lipsum}

\newtheorem{theorem}{Theorem}
\newtheorem{definition}{Definition}

\newtheorem{lemma}{Lemma}
\newtheorem{proposition}{Proposition}
\newtheorem{example}{Example}

\def\bcj{\begin{conjecture}}
	\def\ecj{\end{conjecture}}
\def\bcr{\begin{corollary}}
	\def\ecr{\end{corollary}}
\def\bd{\begin{definition}}
	\def\ed{\end{definition}}
\def\bea{\begin{eqnarray}}
	\def\eea{\end{eqnarray}}
\def\bem{\begin{enumerate}}
	\def\eem{\end{enumerate}}
\def\bex{\begin{example}}
	\def\eex{\end{example}}
\def\bim{\begin{itemize}}
	\def\eim{\end{itemize}}
\def\bl{\begin{lemma}}
	\def\el{\end{lemma}}
\def\bma{\begin{bmatrix}}
	\def\ema{\end{bmatrix}}
\def\bpf{\begin{proof}}
	\def\epf{\end{proof}}
\def\bpp{\begin{proposition}}
	\def\epp{\end{proposition}}
\def\bqu{\begin{question}}
	\def\equ{\end{question}}
\def\br{\begin{remark}}
	\def\er{\end{remark}}
\def\bt{\begin{theorem}}
	\def\et{\end{theorem}}

\def\squareforqed{\hbox{\rlap{$\sqcap$}$\sqcup$}}
\def\qed{\ifmmode\squareforqed\else{\unskip\nobreak\hfil
		\penalty50\hskip1em\null\nobreak\hfil\squareforqed
		\parfillskip=0pt\finalhyphendemerits=0\endgraf}\fi}
\def\endenv{\ifmmode\;\else{\unskip\nobreak\hfil
		\penalty50\hskip1em\null\nobreak\hfil\;
		\parfillskip=0pt\finalhyphendemerits=0\endgraf}\fi}
\newenvironment{proof}{\noindent \textbf{{Proof.~} }}{\qed}
\def\Dbar{\leavevmode\lower.6ex\hbox to 0pt
	{\hskip-.23ex\accent"16\hss}D}
\makeatletter
\def\url@leostyle{%
	\@ifundefined{selectfont}{\def\UrlFont{\sf}}{\def\UrlFont{\small\ttfamily}}}
\makeatother
\def\bcj{\begin{conjecture}}
	\def\ecj{\end{conjecture}}
\def\bcr{\begin{corollary}}
	\def\ecr{\end{corollary}}
\def\bd{\begin{definition}}
	\def\ed{\end{definition}}
\def\bea{\begin{eqnarray}}
	\def\eea{\end{eqnarray}}
\def\bem{\begin{enumerate}}
	\def\eem{\end{enumerate}}
\def\bex{\begin{example}}
	\def\eex{\end{example}}
\def\bim{\begin{itemize}}
	\def\eim{\end{itemize}}
\def\bl{\begin{lemma}}
	\def\el{\end{lemma}}
\def\bpf{\begin{proof}}
	\def\epf{\end{proof}}
\def\bpp{\begin{proposition}}
	\def\epp{\end{proposition}}
\def\bqu{\begin{question}}
	\def\equ{\end{question}}
\def\br{\begin{remark}}
	\def\er{\end{remark}}
\def\bt{\begin{theorem}}
	\def\et{\end{theorem}}

\def\btb{\begin{tabular}}
	\def\etb{\end{tabular}}

\newcommand{\nc}{\newcommand}

\nc{\bbA}{\mathbb{A}} \nc{\bbB}{\mathbb{B}} \nc{\bbC}{\mathbb{C}}
\nc{\bbD}{\mathbb{D}} \nc{\bbE}{\mathbb{E}} \nc{\bbF}{\mathbb{F}}
\nc{\bbG}{\mathbb{G}} \nc{\bbH}{\mathbb{H}} \nc{\bbI}{\mathbb{I}}
\nc{\bbJ}{\mathbb{J}} \nc{\bbK}{\mathbb{K}} \nc{\bbL}{\mathbb{L}}
\nc{\bbM}{\mathbb{M}} \nc{\bbN}{\mathbb{N}} \nc{\bbO}{\mathbb{O}}
\nc{\bbP}{\mathbb{P}} \nc{\bbQ}{\mathbb{Q}} \nc{\bbR}{\mathbb{R}}
\nc{\bbS}{\mathbb{S}} \nc{\bbT}{\mathbb{T}} \nc{\bbU}{\mathbb{U}}
\nc{\bbV}{\mathbb{V}} \nc{\bbW}{\mathbb{W}} \nc{\bbX}{\mathbb{X}}
\nc{\bbZ}{\mathbb{Z}}

\nc{\bA}{{\bf A}} \nc{\bB}{{\bf B}} \nc{\bC}{{\bf C}}
\nc{\bD}{{\bf D}} \nc{\bE}{{\bf E}} \nc{\bF}{{\bf F}}
\nc{\bG}{{\bf G}} \nc{\bH}{{\bf H}} \nc{\bI}{{\bf I}}
\nc{\bJ}{{\bf J}} \nc{\bK}{{\bf K}} \nc{\bL}{{\bf L}}
\nc{\bM}{{\bf M}} \nc{\bN}{{\bf N}} \nc{\bO}{{\bf O}}
\nc{\bP}{{\bf P}} \nc{\bQ}{{\bf Q}} \nc{\bR}{{\bf R}}
\nc{\bS}{{\bf S}} \nc{\bT}{{\bf T}} \nc{\bU}{{\bf U}}
\nc{\bV}{{\bf V}} \nc{\bW}{{\bf W}} \nc{\bX}{{\bf X}}
\nc{\ba}{{\bf a}} \nc{\be}{{\bf e}} \nc{\bu}{{\bf u}}
\nc{\brr}{{\bf r}}

\nc{\cA}{{\cal A}} \nc{\cB}{{\cal B}} \nc{\cC}{{\cal C}}
\nc{\cD}{{\cal D}} \nc{\cE}{{\cal E}} \nc{\cF}{{\cal F}}
\nc{\cG}{{\cal G}} \nc{\cH}{{\cal H}} \nc{\cI}{{\cal I}}
\nc{\cJ}{{\cal J}} \nc{\cK}{{\cal K}} \nc{\cL}{{\cal L}}
\nc{\cM}{{\cal M}} \nc{\cN}{{\cal N}} \nc{\cO}{{\cal O}}
\nc{\cP}{{\cal P}} \nc{\cQ}{{\cal Q}} \nc{\cR}{{\cal R}}
\nc{\cS}{{\cal S}} \nc{\cT}{{\cal T}} \nc{\cU}{{\cal U}}
\nc{\cV}{{\cal V}} \nc{\cW}{{\cal W}} \nc{\cX}{{\cal X}}
\nc{\cZ}{{\cal Z}}

\nc{\hA}{{\hat{A}}} \nc{\hB}{{\hat{B}}} \nc{\hC}{{\hat{C}}}
\nc{\hD}{{\hat{D}}} \nc{\hE}{{\hat{E}}} \nc{\hF}{{\hat{F}}}
\nc{\hG}{{\hat{G}}} \nc{\hH}{{\hat{H}}} \nc{\hI}{{\hat{I}}}
\nc{\hJ}{{\hat{J}}} \nc{\hK}{{\hat{K}}} \nc{\hL}{{\hat{L}}}
\nc{\hM}{{\hat{M}}} \nc{\hN}{{\hat{N}}} \nc{\hO}{{\hat{O}}}
\nc{\hP}{{\hat{P}}} \nc{\hR}{{\hat{R}}} \nc{\hS}{{\hat{S}}}
\nc{\hT}{{\hat{T}}} \nc{\hU}{{\hat{U}}} \nc{\hV}{{\hat{V}}}
\nc{\hW}{{\hat{W}}} \nc{\hX}{{\hat{X}}} \nc{\hZ}{{\hat{Z}}}

\nc{\hn}{{\hat{n}}}

\def\min{\mathop{\rm min}}

\newcommand{\ket}[1]{|#1\rangle}

\usepackage[
colorlinks,
linkcolor = blue,
citecolor = blue,
urlcolor = blue]{hyperref}
\def \qed {\hfill \vrule height7pt width 7pt depth 0pt}

\newcounter{lastnote}

\begin{document}
	
	\title{ Graph connectivity based  strong quantum nonlocality with genuine entanglement   }
	\author{Yan-Ling Wang}
	\email{wangylmath@yahoo.com}
	\affiliation{ School of Computer Science and Techonology, Dongguan University of Technology, Dongguan, 523808, China}
	
	\author{Mao-Sheng Li}
	\email{li.maosheng.math@gmail.com}
	\affiliation{Department of Physics, Southern University of Science and Technology, Shenzhen, 518055, China}
	\affiliation{ Department of Physics, University of Science and Technology of China, Hefei, 230026, China}

	\author{Man-Hong Yung}
	\email{yung@sustc.edu.cn}
	\affiliation{Department of Physics, Southern University of Science and Technology, Shenzhen, 518055, China}
	\affiliation{Institute for Quantum Science and Engineering, and Department of Physics,
		Southern University of Science and Technology, Shenzhen, 518055, China}
	

	\begin{abstract}
		Strong nonlocality based on  local distinguishability is  a    stronger form of quantum nonlocality recently introduced  in multipartite quantum systems:  an orthogonal set of  multipartite  quantum states is said to be  of strong nonlocality if it is locally irreducible for every bipartition of the subsystems.  Most of the known  results   are limited to sets with product states. Shi et al. presented the first result of strongly nonlocal entangled sets  in [\href{https://journals.aps.org/pra/abstract/10.1103/PhysRevA.102.042202}{Phys. Rev. A \textbf{102},  042202 (2020)}] and there they questioned the existence of  strongly nonlocal  set with genuine entanglement. In this work, we  relate the  strong nonlocality of  some speical  set of  genuine entanglement  to  the connectivities of some  graphs. Using this relation, we successfully construct   sets  of genuinely entangled states with strong nonlocality.  As a consequence,  our constructions give a negative  answer to  Shi et al.'s question, which also provide another answer to the open   problem raised by Halder et al. [\href{https://journals.aps.org/prl/abstract/10.1103/PhysRevLett.122.040403}{Phys. Rev. Lett. \textbf{122}, 040403 (2019)}]. This work associates a physical quantity named strong nonlocality with a mathematical quantity called graph connectivity.
		
	\end{abstract}

	\maketitle
	\section{Introduction}
	Quantum nonlocality,  one of the most   surprising  property in  quantum mechanics,  is usually being detected  with entangled states by their violations of Bell-type
	inequalities. In addition, the local indistinguishability  of an  orthogonal  set of quantum states is also being widely used to illustrate the phenomenon of quantum nonlocality.  It is well known that an orthogonal set of quantum states can be perfectly distinguished by positive operation value measurement (POVM)\cite{nils}.   Bennett et al. \cite{Ben99} presented an example
	of  orthogonal product states in $\mathbb{C}^3\otimes  \mathbb{C}^3$ that are locally indistinguishable (here only local operations and classical communications are allowed).  Therefore, all the information of the given set can be inferred by using global measure, but only   partial   information can be deduced when only local measurement are allowed.   They named such a phenomenon as quantum nonlocality without entanglement. Since then, the quantum nonlocality based on local indistinguishability    has been studied extensively (see Refs. \cite{Gho01,Wal00,Wal02,Fan04,Nathanson05,Cohen07,Bandyopadhyay11,Li15,Fan07,Yu12,Cos13,Yu115,Wang19,Xiong19,Li20,Ran04,Hor03,Ben99b,DiVincenzo03,Zhang14,Zhang15,Zhang16,Xu16b,Xu16m,Zhang16b,Wang15,Feng09,
		Yang13,Zhang17,Halder18,Li18,Halder1909,Xu20b,Rout1909}  for an
	incomplete list).   Moreover, the local indistinguishability of  quantum states has also been practically applied in quantum cryptography primitives such as   data hiding \cite{Terhal01,DiVincenzo02} and secret sharing \cite{Markham08,Rahaman15,WangJ17}.

	In  each protocol that can perfectly distinguished the set of states, the states remain to be orthogonal to each other after every local measurement.  Measurements with such property are called orthogonality preserving measurement.   Based on this kind of measurement, Halder \emph{et al.} \cite{Halder19} introduced the concept local irreducibility, a  stronger form of local indistinguishability.   A set of multipartite orthogonal quantum states is said to be locally irreducible if it is not possible to locally eliminate one or more states from that set using orthogonality  preserving  measurement.  They presented   two sets of product states in $\mathbb{C}^3\otimes\mathbb{C}^3 \otimes\mathbb{C}^3$ and $\mathbb{C}^4\otimes\mathbb{C}^4  \otimes\mathbb{C}^4$ that are locally irreducible for each biparition of the corresponding tripartite systems. They named such phenomenon as strong nonlocality without entanglement.  After that, Yuan et al. \cite{Tian20}
	constructed some strongly nonlocal  sets without entanglement in higher dimensional systems $\mathbb{C}^d\otimes \mathbb{C}^d\otimes \mathbb{C}^d$ and even an example in $\mathbb{C}^3\otimes\mathbb{C}^3\otimes \mathbb{C}^3\otimes \mathbb{C}^3.$   Recently, Shi et al. \cite{Shi21} developed a strong method which helps them to  construct some strongly nonlocal orthogonal product sets in general 3,4,5-parties systems.  Based on the local irreducibility in some multipartitions, Zhang \emph{et al.} \cite{Zhang1906} generalized the concept of strong nonlocality to more general settings.     
	
	Intuitively, the more entanglement of a given set, the easier it is  to show the strong  nonlocality. However,  this  may not be the case. In fact,  Halder et al.   found examples of strong nonlocality without entanglement but they proposed the open question: are there  any orthogonal entangled bases that present the strong nonlocality (see Ref. \cite{Halder19})? Moreover, as a negative example, they showed that a special three-qubit GHZ basis   is locally reducible in all bipartitions. Soon after that,  Shi et al. \cite{Shi20S} provided a positive answer to the above question by constructing sets of strongly nonlocal entangled states (which are separable in some partition of the tripartite systems) in $\mathbb{C}^d\otimes\mathbb{C}^d\otimes\mathbb{C}^d$ for $d\geq 3$. But they doubted the existence of genuinely entangled set which is still strongly nonlocal. Therefore, it is interesting to consider whether genuinely entangled set presents this kind of strong nonlocality or not. 
	
	In this paper, we construct genuinely entangled sets with strong nonlocality in $\mathbb{C}^d \otimes\mathbb{C}^d  \otimes\mathbb{C}^d$ for $d\geq 3$. First, we give an example in $\mathbb{C}^3 \otimes\mathbb{C}^3  \otimes\mathbb{C}^3$. Then by relating some graphs to some special genuinely entangled set, we prove that the connectivities of these graphs are sufficient to show the strong nonlocality of this given set.  As a consequence, we construct a genuinely entangled set with $d^3-(d-2)^2$ (or $d^3-(d-2)^2$) elements in $\mathbb{C}^d \otimes\mathbb{C}^d \otimes \mathbb{C}^d$ when $d\geq 3$ is odd (or is even) that is of strong nonlocality.
	
	The rest of this article is organized as follows. In Sec. \ref{sec:second}, we give some necessary notation and definitions. In Sec. \ref{second}, we present a general method to construct strong nonlocal set which is consisting of GHZ like states.  In Sec. \ref{third}, we extend this construction to more general case and present an example.    Finally, we draw a conclusion and present some interesting problems in     section \ref{fifth}.

	\section{Preliminaries }\label{sec:second}
	For an integer $d\geq 2$, denote $\mathbb{Z}_d$ to be the group defined over $\{0,1,\cdots,d-1\}$ with $\mathrm{ mod } \ d $ as its ``$+$" operation.  Let $\cH$ be a Hilbert space of dimension $d$. We alway assume that its computational basis is $\{|i\rangle \mid i\in \mathbb{Z}_d\}.$
	
	In this section, we will   introduce the definition of the graph and some related usages, the concept of genuine entanglement and the definition of strongest nonlocality.
	
	\vskip 5pt
	
	\noindent\emph{Graph.}-- A graph (see Ref. \cite{Reinhard05} for more details)   is a pair $G = (V, E)$, where $V$ is a set whose elements are called vertices   and $E$ is a set of paired vertices, whose elements are called edges.  
	
	In a  graph, a  pair of vertices $\{x, y\}$ (where $x,y\in V$) is called connected if there is  a path  with edges from $x$ to $y.$ Otherwise, the   pair is called disconnected. A connected graph is a  graph in which every   pair of vertices in the graph is connected. There is an important concept called  connected component which is related to the connectivity of a graph. A connected component is a maximal connected subgraph of a  graph.  Then a graph is connected if and only if it has exactly one connected component.
	
	\vskip 5pt

	\noindent\emph{Genuine entanglement.}-- In a bipartite system $\cH_A\otimes \cH_B$, a pure $|\Psi\rangle_{AB}$ is called a entangled state if it can not write as some tensor product of two local pure states,  i.e., $|\Psi\rangle_{AB}$ is not of the form $|\phi\rangle_A\otimes|\theta\rangle_B$. For a pure state $|\Psi\rangle_{A_1\cdots A_n}$ in  multipartite systems $\otimes_{i=1}^n \cH_{A_i}$, it is called a genuinely entangled state  if it is entangled for each bipartition  of $\{A_1,\cdots,A_n\}$ (see Ref. \cite{Markiewicz13}). 
	
	The most well known genuinely entangled state is  the Greenberger-Horne-Zeilinger (GHZ) state $(|000\rangle+|111\rangle)/\sqrt{2}$ and the $W$ state $(|100\rangle+|010\rangle+|001\rangle)/\sqrt{3}$ in three qubits. For a higher dimensional systems   $\mathbb{C}^{d_1}\otimes \mathbb{C}^{d_2}\otimes  \mathbb{C}^{d_3}$, the states $(|i_1j_1k_1\rangle\pm|i_2j_2k_2\rangle)/\sqrt{2}$  are also genuinely entangled if $i_1\neq i_2, j_1\neq j_2,$ and $k_1\neq k_2.$ We call such states as GHZ like states under the computational basis. In addition,  the state $$(|i_1j_1k_1\rangle+w|i_2j_2k_2\rangle+w^2|i_3j_3k_3\rangle +\cdots+w^{d-1}|i_dj_dk_d\rangle)/\sqrt{d}$$ 
	(where $w$ is any $d$-th root of unit and $i_m\in \mathbb{Z}_{d_1},$  $j_m \in \mathbb{Z}_{d_2},$ $k_m\in \mathbb{Z}_{d_3}$) is also genuinely entangled if  $i_m \neq i_{m'}$  $j_m \neq j_{m'}$ and  $k_m\neq k_{m'}$ whenever   $m\neq m'$.  We call these states   GHZ like states with  weight $d$. And  the set  $\{|i_mj_mk_m\rangle\}_{m=1}^d$ with the  property,  $i_m \neq i_{m'}$  $j_m \neq j_{m'}$ and  $k_m\neq k_{m'}$ whenever   $m\neq m'$,  is  called coordinately different. One should note that the GHZ like states with weight $d$ in $\mathbb{C}^{d}\otimes \mathbb{C}^{d}\otimes  \mathbb{C}^{d}$ are 1-uniform states \cite{Goyeneche14} or absolute maximally entangled states \cite{Facchi08, Arnaud13,Huber18}.
	In this paper, we mainly pay attention to the above  form of genuinely entangled states.
	
	\vskip 5pt 
	
	\noindent\emph{Strongest nonlocality.}--  A measurement is nontrivial if not all the
	POVM elements are proportional to the identity operator.  Otherwise, the measurement is trivial.  An set  of orthogonal states is said to be of   the strongest nonlocality  if only trivial  orthogonality-preserving POVM can perform for each bipartition of the subsystems.
	
	Note that a set of the strongest nonlocality must also be of strong nonlocality by definition. In this paper, we mainly pay attention to the strongest nonlocality of some  given   orthogonal genuinely entangled set (OGES). 
	
	\section{Strongly nonlocal sets of     GHZ like states   }\label{second}
	
	In this section,  we show that  genuinely entangled set that   are of the strongest nonlocality do exist. Moreover, the strongest nonlocality of some special sets are rather simple as it is determined by the connectivity of some related graphs. First, we present an example to show that orthogonal set of genuinely entangled states with the strongest nonlocality do exist even in $\mathbb{C}^3\otimes \mathbb{C}^3\otimes \mathbb{C}^3.$
	
	\begin{figure}[h]
		\centering		\includegraphics[width=0.5\textwidth,height=0.43\textwidth]{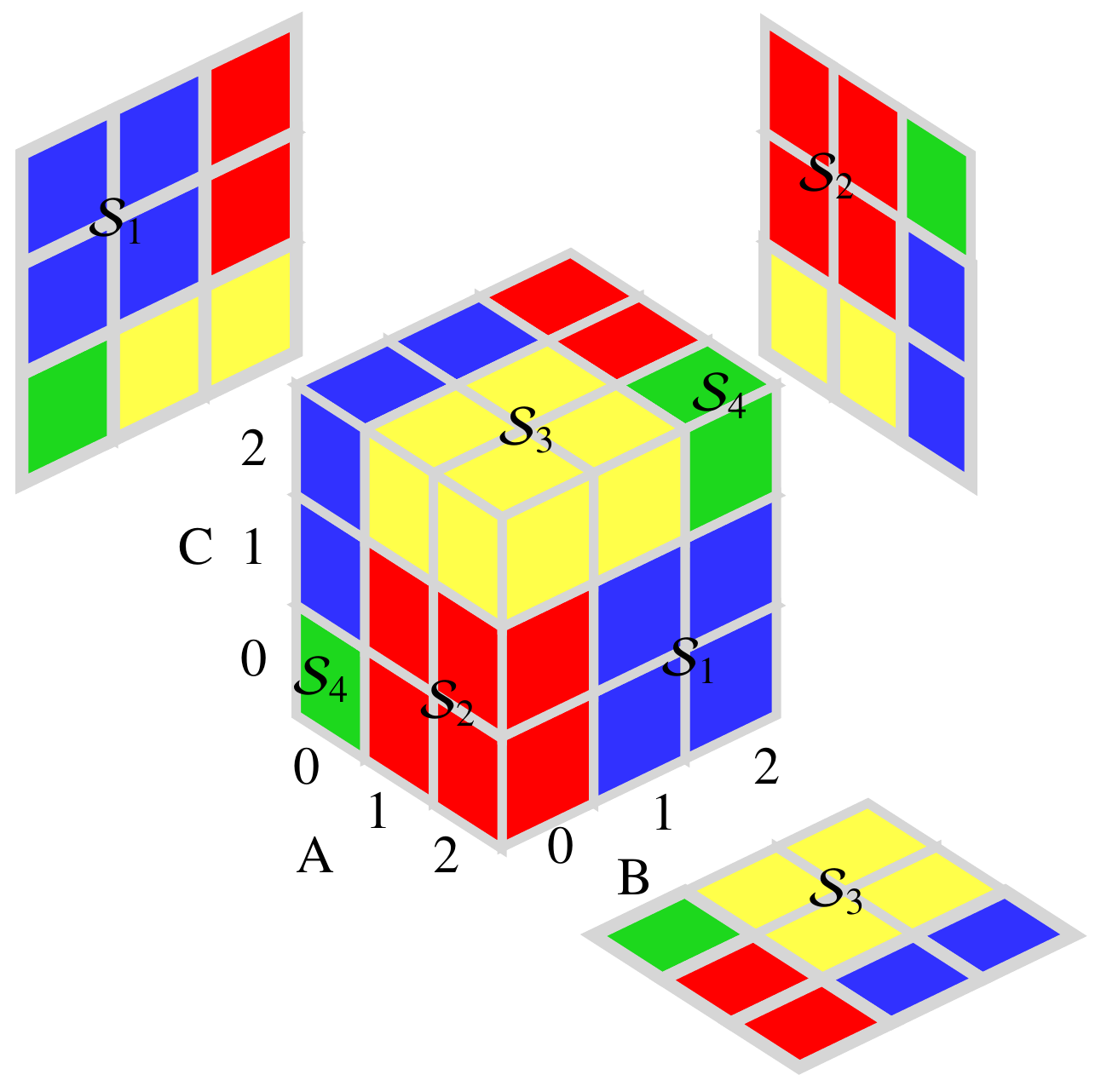}
		\caption{This shows the geometric structure of $\cS_i$ defined in Eq.~\eqref{OGES333}. }\label{fig:333cube}
	\end{figure}

	\begin{example}\label{OEGS_SN_333}
		In $\mathbb{C}^3\otimes \mathbb{C}^3\otimes \mathbb{C}^3$, the set
		$\cS:=\cup_{i=1}^{4}\cS_i$  given by Eq. \eqref{OGES333} is  an  \emph{OGES} of  the strongest nonlocality. The size of this set is $26$ (the geometric structure of $\cS$ can be  seen in Fig. \ref{fig:333cube}).
	\end{example}
	{\small	\begin{equation}\label{OGES333}
			\begin{aligned}
				\cS_1&:=\{\ket{0}_A\ket{i}_B\ket{j+1}_C \pm \ket{2}_A\ket{i+1}_B\ket{j}_C \mid (i,j)\in  \mathbb{Z}_2\times   \mathbb{Z}_2  \},\\
				\cS_2&:=\{\ket{i+1}_A\ket{0}_B\ket{j}_C \pm \ket{i}_A\ket{2}_B\ket{j+1}_C \mid (i,j)\in   \mathbb{Z}_2\times   \mathbb{Z}_2  \},\\	
				\cS_3&:=\{\ket{i}_A\ket{j+1}_B\ket{0}_C \pm \ket{i+1}_A\ket{j}_B\ket{2}_C \mid (i,j)\in   \mathbb{Z}_2\times   \mathbb{Z}_2  \},\\
				\cS_4&:=\{\ket{0}_A\ket{0}_B\ket{0}_C \pm \ket{2}_A\ket{2}_B\ket{2}_C    \}. \\		
			\end{aligned}
		\end{equation}
	}
	\begin{proof}
		Without loss of generality, let $B$ and $C$ come together to  perform a joint  orthogonality-preserving POVM $\{E=M^{\dagger}M\}$, where $E=(a_{ij,k\ell})_{i,j,k,\ell\in \mathbb{Z}_3}$. Then the postmeasurement states $\{\mathbb{I}_A\otimes M\ket{\psi}\big |\ket{\psi}\in \cS\}$ should be mutually orthogonal, i.e., 
		\begin{equation}\label{eq:OGES333relation}
			\langle \phi|\mathbb{I}_A\otimes E |\psi\rangle=0
		\end{equation}  
		for $|\psi\rangle, |\phi\rangle \in \cS$ and  $|\psi\rangle\neq |\phi\rangle.$
		
		First, we  show that the matrix $E$ is diagonal.    For any  pair of non-equal  coordinates  $(i,j)$ and $(k,l)$  in $ \mathbb{Z}_3\times   \mathbb{Z}_3$,  one finds that there exist  two pairs of genuinely entangled states $ \{|\psi_\pm\rangle:=|0\rangle_A|i\rangle_B|j\rangle_C\pm |a_1\rangle_A|b_1\rangle_B|c_1\rangle_C\}$  and  $\{|\phi_\pm\rangle:=|0\rangle_A|k\rangle_B|l\rangle_C\pm |a_2\rangle_A|b_2\rangle_B|c_2\rangle_C\}$ (maybe $|0\rangle_A|i\rangle_B|j\rangle_C$ or  $|0\rangle_A|k\rangle_B|l\rangle_C$ is in the  second term) in $\cS$. Replacing $|\psi\rangle$ by one of   $|\psi_\pm\rangle$ and $|\phi\rangle$ by one of   $|\phi_\pm\rangle$ in Eq.~\eqref{eq:OGES333relation}, we get four equations. Using these four  relations, one obtain that
		$$a_{ij,kl}= {}_A\langle 0|{}_B\langle i|{}_C\langle j| \mathbb{I}_A\otimes E   |0\rangle_A|k\rangle_B|l\rangle_C=0.$$
		Therefore, the matrix $E$ is diagonal under the computational basis.	
		
		In the following, we show that $E$ is indeed proportional to the identity operator. Note the following obeservation: appling   Eq.~\eqref{eq:OGES333relation} to any pair of  states in $\cS$ of the form	$|i_1\rangle_A|j_1\rangle_B|k_1\rangle_C\pm |i_{2}\rangle_A|j_{2}\rangle_B|k_{2}\rangle_C,$ as $i_1\neq i_2$, one could deduce that
		\begin{equation}\label{eq:OGES333Equal}
			a_{j_1k_1,j_1k_1}=a_{j_2k_2,j_2k_2}.
		\end{equation}
		
		Using this obeservation to the following pairs:
		{\small$$
			\begin{array}{ll}
				|1\rangle_A|0\rangle_B|0\rangle_C\pm |0\rangle_A|2\rangle_B|1\rangle_C, &	|0\rangle_A|1\rangle_B|2\rangle_C\pm |2\rangle_A|2\rangle_B|1\rangle_C,\\
				|0\rangle_A|2\rangle_B|0\rangle_C\pm |1\rangle_A|1\rangle_B|2\rangle_C, &	|0\rangle_A|1\rangle_B|1\rangle_C\pm |2\rangle_A|2\rangle_B|0\rangle_C,\\
				|0\rangle_A|0\rangle_B|2\rangle_C\pm |2\rangle_A|1\rangle_B|1\rangle_C,& 
				|1\rangle_A|1\rangle_B|0\rangle_C\pm |2\rangle_A|0\rangle_B|2\rangle_C, \\
				|0\rangle_A|0\rangle_B|1\rangle_C\pm |2\rangle_A|1\rangle_B|0\rangle_C,& 
				|1\rangle_A|0\rangle_B|1\rangle_C\pm |0\rangle_A|2\rangle_B|2\rangle_C, \\	
			\end{array}
			$$
		}
		we have the equations
		$$
		\begin{array}{ll}
			a_{00,00}=a_{21,21},	 &	 a_{21,21}=a_{12,12},\\
			a_{12,12}=a_{20,20}, & a_{20,20}=a_{11,11},	\\
			a_{11,11}=a_{02,02},& 
			a_{02,02}=a_{10,10}, \\
			a_{10,10}=a_{01,01},& 
			a_{01,01}=a_{22,22}.\\	
		\end{array}
		$$
		Therefore, one concludes  that	$E$ is proportional to the identity operator.	 
	\end{proof}
	
	\vskip 5pt
	
	To prove that the elements on the diagonal of $E$  are equal to each other, one can also attach the set $\cS$ with a graph $\cG_A(\cS)=(V_A(\cS),E_A(\cS))$ defined as follows: $V_A(\cS):=\mathbb{Z}_3\times \mathbb{Z}_3$, $\{(j_1,k_1),(j_2,k_2)\}\in E_A(\cS)$ if and only if there exist some $i_1,i_2\in \mathbb{Z}_3$ such that  the pair of states 
	$|i_1\rangle_A|j_1\rangle_B|k_1\rangle_C\pm |i_{2}\rangle_A|j_{2}\rangle_B|k_{2}\rangle_C $
	are in $\cS.$  Under this definition, by Eq.~\eqref{eq:OGES333Equal}, if $\cG_A(\cS)$ is connected, then the elements on the diagonal of  $E$ must be equal to each other. And the graph $\cG_A(\cS)$ can be seen in the Fig. \ref{fig:OEGS333}. One can easily check that it is connected. Therefore, one can also conclude that $E$ is proportional to identity operator via the connectivity of this graph.
	
	\vskip 5pt

	\begin{figure}[h]
		\centering
		\includegraphics[scale=0.32]{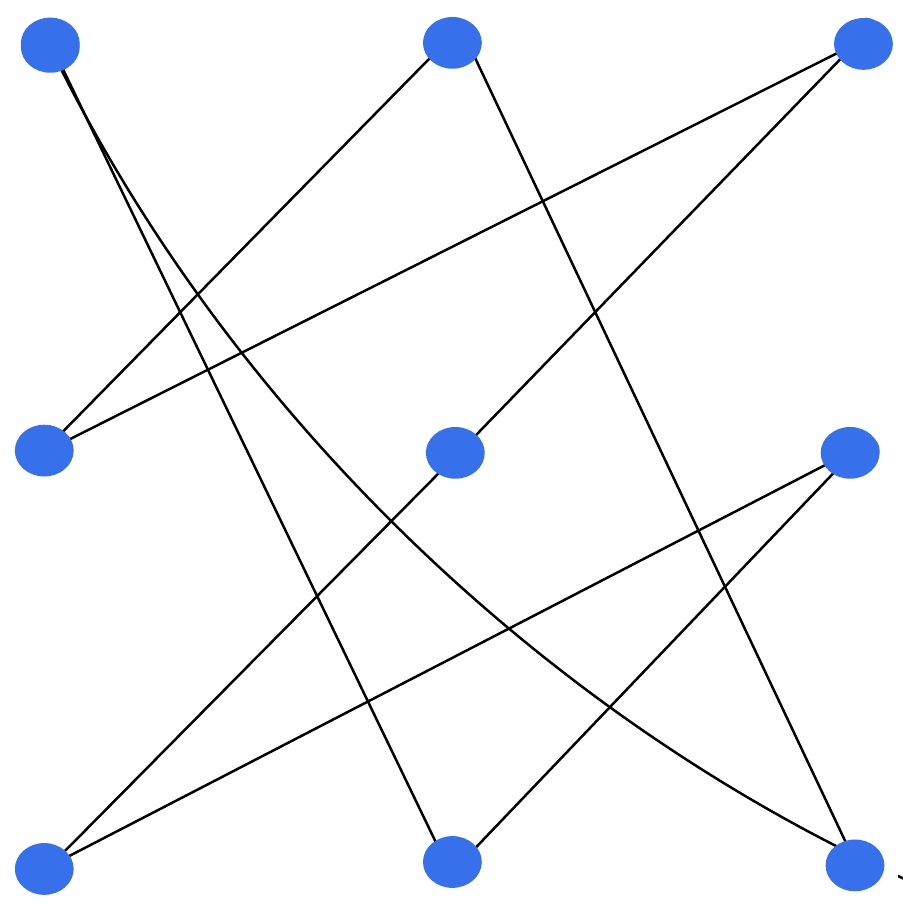}
		\caption{This shows the $\cG_A(\cS)$ corresponding  to the OGES  of  $\mathbb{C}^3\otimes \mathbb{C}^3\otimes \mathbb{C}^3$ defined by Eq.~\eqref{OGES333}. }\label{fig:OEGS333}
	\end{figure}

	To prove that the off diagonal elements of $E$ are all zero, we use the following two facts:
	\begin{enumerate}[(a)] 
		\item The set $\cS$ can be separated into pairs of GHZ like states of the form 
		$|i_1\rangle_A|j_1\rangle_B|k_1\rangle_C\pm |i_{2}\rangle_A|j_{2}\rangle_B|k_{2}\rangle_C.$
		\item The  set $\cC(\cS)$ defined by 
		$\{(i_1,j_1,k_1),(i_2,j_2,k_2)\mid |i_1\rangle_A|j_1\rangle_B|k_1\rangle_C\pm |i_{2}\rangle_A|j_{2}\rangle_B|k_{2}\rangle_C\in \cS \}$
		contains the subsets $\{0\}\times \mathbb{Z}_{3}\times\mathbb{Z}_{3},$   $\mathbb{Z}_{3}\times \{0\}\times\mathbb{Z}_{3}$  and  $\mathbb{Z}_{3}\times \mathbb{Z}_{3}\times\{0\}$ of  $\mathbb{Z}_{3}\times \mathbb{Z}_{3}\times\mathbb{Z}_{3}.$
		
	\end{enumerate}
	
	In order to move to the general tripartite systems, we make two assumptions on the given set   $\cS\subseteq\mathbb{C}^{d_1}\otimes \mathbb{C}^{d_2}\otimes \mathbb{C}^{d_3}$ to be distinguished. The first one  is that the set $\cS$ can be separated into pairs of GHZ like states of the form 
	$|i_1\rangle_A|j_1\rangle_B|k_1\rangle_C\pm |i_{2}\rangle_A|j_{2}\rangle_B|k_{2}\rangle_C.$
	The second one is that the  set $\cC(\cS)$ defined by 
	$\{(i_1,j_1,k_1),(i_2,j_2,k_2)\mid |i_1\rangle_A|j_1\rangle_B|k_1\rangle_C\pm |i_{2}\rangle_A|j_{2}\rangle_B|k_{2}\rangle_C\in \cS \}$
	contains the subsets $\{i_0\}\times \mathbb{Z}_{d_2}\times\mathbb{Z}_{d_3},$   $\mathbb{Z}_{d_1}\times \{j_0\}\times\mathbb{Z}_{d_3}$  and  $\mathbb{Z}_{d_1}\times \mathbb{Z}_{d_2}\times\{k_0\}$ for some $i_0\in \mathbb{Z}_{d_1}, j_0\in \mathbb{Z}_{d_2}, k_0\in \mathbb{Z}_{d_3}.$ If the first condition is satisfied, the  set $\cS$ is called a special set of GHZ like states under the computational basis  And we call the set  $\cS$  is plane containing if it satisfies the second condition.

	For each partition $A|BC,B|CA,$ and $C|AB$, we attach the set $\cS$ with a graph $\cG_{A}(\cS)$, $\cG_{B}(\cS)$ and $\cG_{C}(\cS)$ respectively. Here we give the exact description of $\cG_B(\cS)$ as example, $\cG_{A}(\cS)$ and $\cG_{C}(\cS)$ can be defined similarly.  The graph  $\cG_B(\cS)$ is determined by its vertexes $V_B(\cS)$  and  edges $E_B(\cS).$ Here $V_B(\cS)$ is defined to be $\mathbb{Z}_{d_3}\times \mathbb{Z}_{d_1}.$ A pair of nodes $\{(k_1,i_1),(k_2,i_2)\}$ belongs to $E_B(\cS)$ if and only if there exist some $j_1,j_2\in \mathbb{Z}_{d_2}$ such that  the pair of states 
	$|i_1\rangle_A|j_1\rangle_B|k_1\rangle_C\pm |i_{2}\rangle_A|j_{2}\rangle_B|k_{2}\rangle_C $
	are in $\cS.$ We call the graph  $\cG_B(\cS)$  the corresponding graph  of $\cS$ with respect to the partition $B|CA$.

	\begin{theorem}\label{GHZ_SN_tri}
		Let $\cS$ be a special orthogonal  set of GHZ like states  under computational basis in $\mathbb{C}^{d_1}\otimes\mathbb{C}^{d_2}\otimes \mathbb{C}^{d_3}$ and suppose that it is plane containing. Then the set  $\cS$ is of the strongest nonlocality if and only if all  the    corresponding graphs of  $\cS$    are connected.
	\end{theorem}	
	\begin{proof}($\Leftarrow$) 
		Without loss of generality, suppose that  $B$ and $C$ come together and   perform a joint  orthogonality-preserving POVM $\{E=M^{\dagger}M\}$, where $E=(a_{ij,k\ell})_{i,k\in \mathbb{Z}_{d_2},j,\ell\in \mathbb{Z}_{d_3}}$. The postmeasurement states $\{\mathbb{I}_A\otimes M\ket{\psi}\big |\ket{\psi}\in \cS\}$ should be mutually orthogonal,  i.e., 
		\begin{equation}\label{eq:OGESgenrelation}
			\langle \phi|\mathbb{I}_A\otimes E |\psi\rangle=0
		\end{equation}  
		for $|\psi\rangle, |\phi\rangle \in \cS$ and  $|\psi\rangle\neq |\phi\rangle.$  First, we  show that the matrix $E$ is diagonal.    For any  pair of non-equal  coordinates  $(i,j)$ and $(k,l)$  in $\mathbb{Z}_{d_2}\times \mathbb{Z}_{d_3}$,  as  $\cS$ is  plane containing,   there exist  two pairs of GHZ like states $ \{|\psi_\pm\rangle:=|i_0\rangle_A|i\rangle_B|j\rangle_C\pm |a_1\rangle_A|b_1\rangle_B|c_1\rangle_C\}$  and  $\{|\phi_\pm\rangle:=|i_0\rangle_A|k\rangle_B|l\rangle_C\pm |a_2\rangle_A|b_2\rangle_B|c_2\rangle_C\}$ (maybe $|i_0\rangle_A|i\rangle_B|j\rangle_C$ or $|i_0\rangle_A|k\rangle_B|l\rangle_C$ is in the second term) in $\cS$. Replacing $|\psi\rangle$ by one of   $|\psi_\pm\rangle$ and $|\phi\rangle$ by one of   $|\phi_\pm\rangle$ in Eq.~\eqref{eq:OGES333relation}, we get four equations. Using these four  relations, one obtain 
		$$a_{ij,kl}= {}_A\langle i_0|{}_B\langle i|{}_C\langle j| \mathbb{I}_A\otimes E   |i_0\rangle_A|k\rangle_B|l\rangle_C=0.$$
		Therefore, the matrix $E$ must be a diagonal matrix  under the computational basis.
		
		In the following, we show that the elements on the diagonal of $E$ are all equal to each other. As the graph $\cG_A(\cS)$ is connected, for any two different nodes  $(j_I,k_I)$ and $(j_L,k_L)$ of $\cG_A(\cS)$, there is a path connecting them, namely	
		\begin{equation}\label{eq:path}
			(j_1,k_1)\overset{e_1}{-} (j_2,k_2)\overset{e_2}{-}\cdots \overset{e_{N-1}}{-} (j_{N},k_{N})\overset{e_N}{-}(j_{N+1},k_{N+1})
		\end{equation}  
		here  $(j_1,k_1)=(j_I,k_I)$ and $(j_{N+1},k_{N+1})=(j_L,k_L)$. By the definition of the edges of $\cG_A(\cS)$, for each edge $e_l=\{(j_l,k_l),(j_{l+1},k_{l+1})\},$ there is a pair of     states of the form $|i_l\rangle_A|j_l\rangle_B|k_l\rangle_C\pm |i_{l+1}\rangle_A|j_{l+1}\rangle_B|k_{l+1}\rangle_C\in\cS$ for some different $i_l,i_{l+1} \in \mathbb{Z}_{d_1}.$	
		Appling Eq.~\eqref{eq:OGESgenrelation}  to each pair of genuinely entangled states $|i_l\rangle_A|j_l\rangle_B|k_l\rangle_C\pm |i_{l+1}\rangle_A|j_{l+1}\rangle_B|k_{l+1}\rangle_C,$	 one can 
		obtain that 
		\begin{equation}\label{eq:generaleq} a_{j_{l}k_l,j_lk_l}=a_{j_{l+1}k_{l+1},j_{l+1}k_{l+1}},\ \ 1\leq l\leq N.
		\end{equation}	
		From the path in Eq. \eqref{eq:path} and Eq.~\eqref{eq:generaleq}, one could easily deduce that 
		$$ a_{j_{I}k_I,j_Ik_I}=a_{j_{L}k_{L},j_{L}k_{L}}.$$
		
		Therefore, one can conclude that the matrix $E$ is indeed proportional to identity operator. 
		
		\vskip 5pt
		
		($\Rightarrow$)	 Without loss of generality, we only need to show that $\cG_A(\cS)$ is connected. By the strongest nonlocality of $\cS$, there is only one solution (which is proportional to the identity operator) for $E$ when 	it satisfies all the equations in Eq.~\eqref{eq:OGESgenrelation}. By appling Eq.~\eqref{eq:OGESgenrelation} to two states $|\phi\rangle $ and $|\psi\rangle$ from different   pair of GHZ like states, we can only obtain some linear equations for the off diagonal elements of $E.$ Therefore, only the following equations  are related to  the diagonal elements of $E$
		$$\{	\langle \phi_+|\mathbb{I}_A\otimes E |\phi_-\rangle=0 \mid  |\phi_\pm\rangle \text{  a pair of GHZ like states in } \cS \}.$$
		For each pair of GHZ like states $|\phi_\pm\rangle=|i_1\rangle_A|j_1\rangle_B|k_1\rangle_C\pm |i_{2}\rangle_A|j_{2}\rangle_B|k_{2}\rangle_C$ in $\cS,$ the corresponding equation in the above set is just
		$ a_{j_1k_1,j_1k_1} =a_{j_2k_2,j_2k_2}.$ This equation exactly corresponds to the edge  connecting $(j_1,k_1)$ and $(j_2,k_2)$ in $\cG_A(S)$. From this correspondence, one can deduce that $\cG_A(S)$ is connected as the diagonal elements of $E$ are equal to each other.
	\end{proof}

	\vskip 5pt 
	From the above theorem, the problem of the  strongest nonlocality of some special orthogonal set with GHZ like states can be determined by the connectivity of its related graphs.	Moreover, one can easily extend the above theorem to multipartite systems. Now we apply Theorem \ref{GHZ_SN_tri} to get some sets of the strong nonlocality with   genuine entanglement.

	\begin{figure}[h]
		\centering
		\includegraphics[scale=0.84]{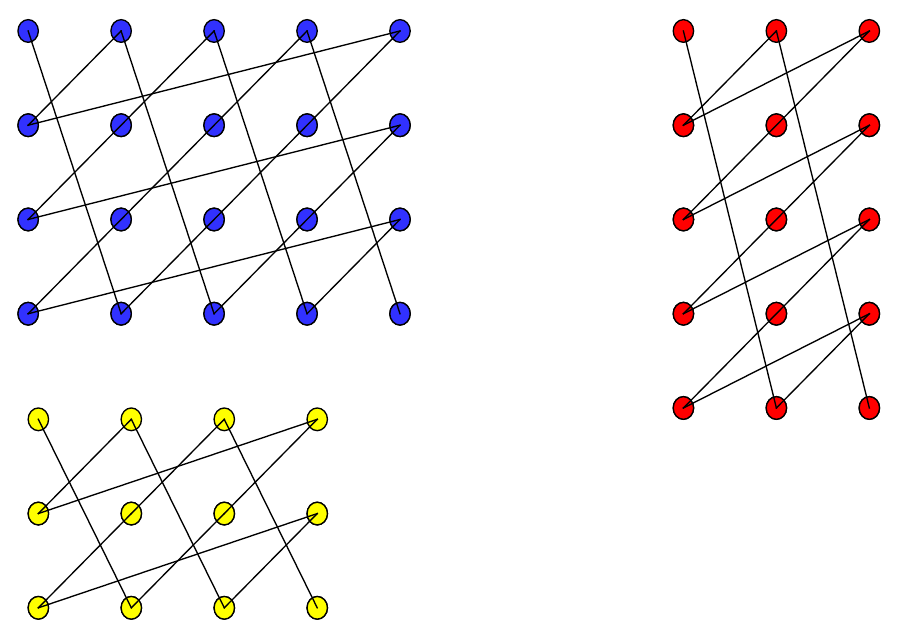}  	\caption{This shows the three graphs ($\cG_A(\cS)$ without the edge $\{(0,0),(3,4)\}$: blue, $\cG_B(\cS)$ without the edge $\{(0,0),(4,2)\}$: red,  $\cG_C(\cS)$ without the edge $\{(0,0),(2,3)\}$: yellow) corresponding  to the OGES of $\mathbb{C}^3\otimes \mathbb{C}^4\otimes \mathbb{C}^5$ defined by  Eq. \eqref{GES345}.    }\label{fig:345cube}
	\end{figure}
	
	\begin{example}\label{OPS_SN_333}
		In $\mathbb{C}^3\otimes \mathbb{C}^4\otimes \mathbb{C}^5$, the set
		$\cS:=\cup_{i=1}^{4}\cS_i$    is  an  \emph{OGES} of  the strongest nonlocality, where
		{\small	\begin{equation}\label{GES345}
				\begin{aligned}
					\cS_1&:=\{\ket{0}_A\ket{i}_B\ket{j+1}_C \pm \ket{2}_A\ket{i+1}_B\ket{j}_C \mid (i,j)\in  \mathbb{Z}_3\times  \mathbb{Z}_4  \},\\
					\cS_2&:=\{\ket{i+1}_A\ket{0}_B\ket{j}_C \pm \ket{i}_A\ket{3}_B\ket{j+1}_C \mid (i,j)\in  \mathbb{Z}_2\times  \mathbb{Z}_4  \},\\	
					\cS_3&:=\{\ket{i}_A\ket{j+1}_B\ket{0}_C \pm \ket{i+1}_A\ket{j}_B\ket{4}_C \mid (i,j)\in  \mathbb{Z}_2\times  \mathbb{Z}_3  \},\\
					\cS_4&:=\{\ket{0}_A\ket{0}_B\ket{0}_C \pm \ket{2}_A\ket{3}_B\ket{4}_C    \}. 		
				\end{aligned}
			\end{equation}
		}
	\end{example}

	In fact,  $\cS$ is a special set of GHZ like states that is plane containing. Moreover, its corresponding graphs $\cG_A(\cS), \cG_B(\cS)$ and $\cG_C(\cS)$ are showed in the   Fig. \ref{fig:345cube} and they are all connected. By Theorem \ref{GHZ_SN_tri}, the set $\cS$ is of the strongest nonlocality.

	\begin{proposition}\label{pro:GES_SN_tri_o}
		In  $\mathbb{C}^d\otimes \mathbb{C}^d\otimes \mathbb{C}^d$, $d\geq 3$ and odd, the set $\cup_{i=1}^4 \cS_i$ given by Eq. (\ref{OGESdddo})  is an  OGES of  the strongest nonlocality. The size of this set is $d^3-(d-2)^3$. Here and the following we use the notation ${\hat{d}} :=d-1$ for simplicity.
	\end{proposition}
	{\small	\begin{equation}\label{OGESdddo}
			\begin{aligned}
				\cS_1&:=\{\ket{0}_A\ket{i}_B\ket{j+1}_C \pm \ket{\hat{d}}_A\ket{i+1}_B\ket{j}_C \mid (i,j)\in  \mathbb{Z}_{\hat{d}} \times  \mathbb{Z}_{\hat{d}}  \},\\
				\cS_2&:=\{\ket{i+1}_A\ket{0}_B\ket{j}_C \pm \ket{i}_A\ket{\hat{d}}_B\ket{j+1}_C \mid (i,j)\in  \mathbb{Z}_{\hat{d}} \times  \mathbb{Z}_{\hat{d}}   \},\\	
				\cS_3&:=\{\ket{i}_A\ket{j+1}_B\ket{0}_C \pm \ket{i+1}_A\ket{j}_B\ket{\hat{d}}_C \mid (i,j)\in  \mathbb{Z}_{\hat{d}} \times  \mathbb{Z}_{\hat{d}}   \},\\
				\cS_4&:=\{\ket{0}_A\ket{0}_B\ket{0}_C \pm \ket{\hat{d}}_A\ket{\hat{d}}_B\ket{\hat{d}}_C    \}. \\		
			\end{aligned}
		\end{equation}
	}
	The  proof of Proposition  \ref{pro:GES_SN_tri_o}   is  given  in  Appendix.

	\begin{proposition}\label{pro:GES_SN_tri_e}
		In  $\mathbb{C}^d\otimes \mathbb{C}^d\otimes \mathbb{C}^d$, $d\geq 4$ and even, the set $\cup_{i=1}^5 \cS_i$ given by Eq. (\ref{OGESddde})  is an  OGES of  the strongest nonlocality. The size of this set is $d^3-(d-2)^3+2$.
	\end{proposition}
	{\small	\begin{equation}\label{OGESddde}
			\begin{aligned}
				\cS_1&:=\{\ket{0}_A\ket{i}_B\ket{j+1}_C \pm \ket{\hat{d}}_A\ket{i+1}_B\ket{j}_C \mid (i,j)\in  \mathbb{Z}_{\hat{d}} \times  \mathbb{Z}_{\hat{d}}  \},\\
\cS_2&:=\{\ket{i+1}_A\ket{0}_B\ket{j}_C \pm \ket{i}_A\ket{\hat{d}}_B\ket{j+1}_C \mid (i,j)\in  \mathbb{Z}_{\hat{d}} \times  \mathbb{Z}_{\hat{d}}   \},\\	
\cS_3&:=\{\ket{i}_A\ket{j+1}_B\ket{0}_C \pm \ket{i+1}_A\ket{j}_B\ket{\hat{d}}_C \mid (i,j)\in  \mathbb{Z}_{\hat{d}} \times  \mathbb{Z}_{\hat{d}}   \},\\
				\cS_4&:=\{\ket{0}_A\ket{0}_B\ket{0}_C \pm \ket{2}_A\ket{3}_B\ket{2}_C  \},
				\\	
				\cS_5&:=\{ \ket{\hat{d}}_A\ket{\hat{d}}_B\ket{\hat{d}}_C \pm \ket{2}_A\ket{3}_B\ket{3}_C\}.	
			\end{aligned}
	\end{equation}}  
	\begin{figure}[h]
		\centering
		\includegraphics[scale=0.315]{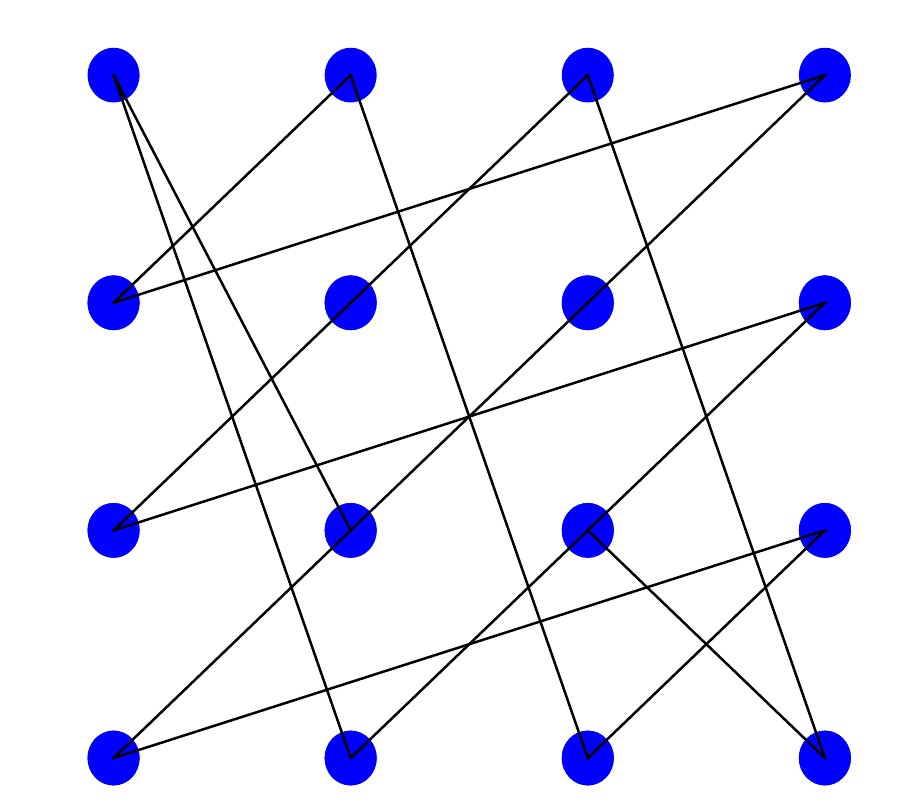}  
		\includegraphics[scale=0.285]{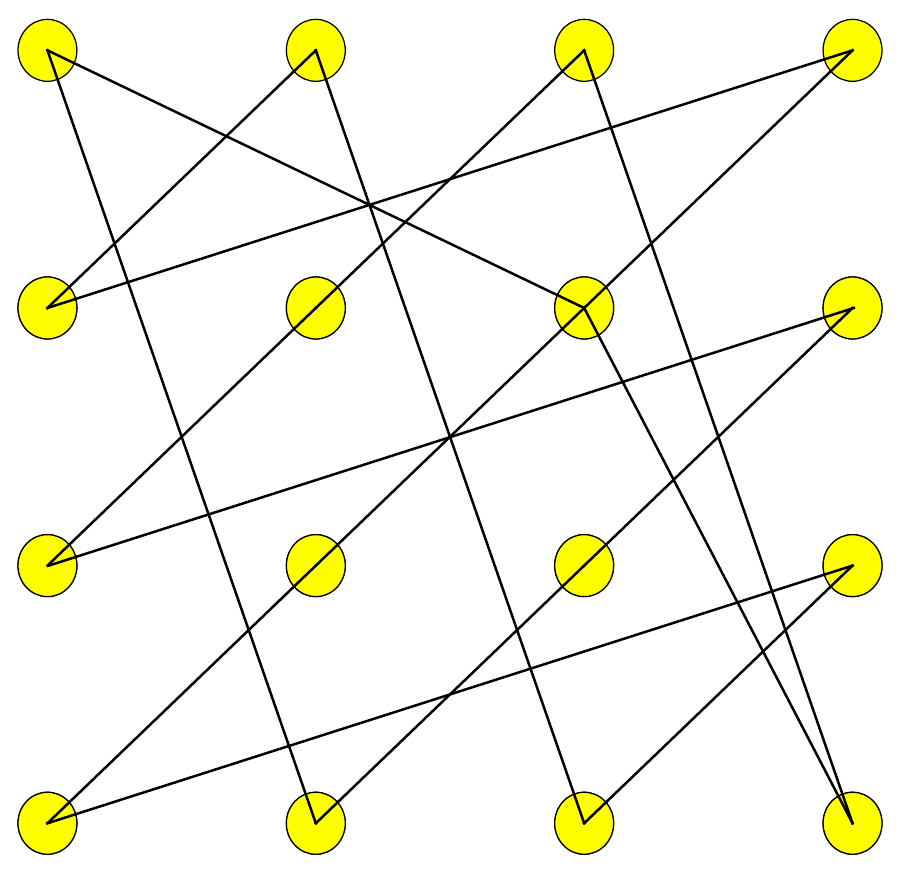}
		\includegraphics[scale=0.315]{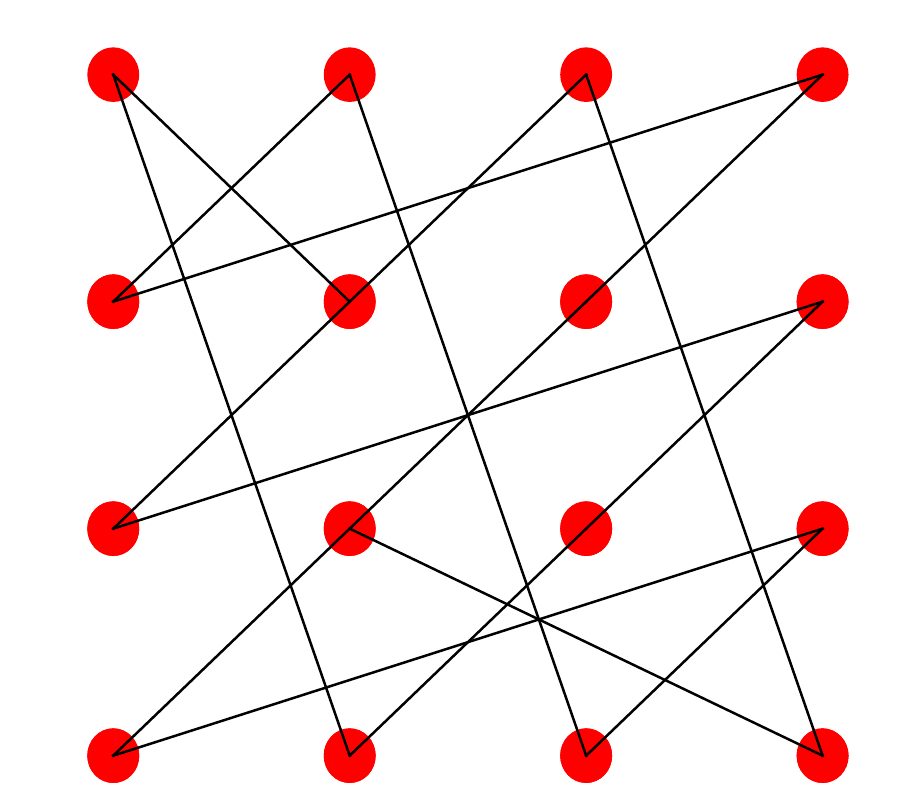}\ 
		\caption{This shows the three graphs  ($\cG_A(\cS)$ blue, $\cG_B(\cS)$ red,  $\cG_C(\cS)$ yellow) corresponds  to the OGES  of $\mathbb{C}^4\otimes \mathbb{C}^4\otimes \mathbb{C}^4$ defined by Eq. \eqref{OGESddde}. }\label{fig:444cube}
	\end{figure}
	
	The  proof of Proposition  \ref{pro:GES_SN_tri_e}   is  given  in  Appendix. Moreover, as an example,  we show the corresponding graphs of $\cS$ for $d=4$  in figure \ref{fig:444cube}.

	\section{Strongly nonlocal sets of high weighted GHZ like states}\label{third}
	In this section, we will consider the strong nonlocality of those sets  of  high weight GHZ like states. We consider a general tripartite systems $\mathbb{C}^{d_1}\otimes \mathbb{C}^{d_2}\otimes \mathbb{C}^{d_3}$ where $d_i\geq 3$  for $1\leq i\leq 3$. We also make some assumptions on the set  $\cS$ to be distinguished. The first one is that the set $\cS$ can be divided   into several $d$-tuples (maybe with different $d$ and here $d\leq \min\{d_1,d_2,d_3\}$) of  GHZ like states with weight $d$ as the following form 
	$$
	\big\{\sum_{m=1}^dw_d^{(m-1)n}|i_mj_mk_m\rangle \mid n\in\mathbb{Z}_d\}$$ 
	where $w_d:=e^{\frac{2 \pi \sqrt{-1}}{d}}.$ Such   $d$-tuples can be also regarded as $d$ linear combinations of $d$ coordinately different vectors $\{|i_mj_mk_m\rangle\}_{m=1}^d$ of the computational basis where the combination coefficients of each state are the elements in the row of the $d$ dimensional Fourier transform  $F_d:=(w_d^{(m-1)(n-1)})_{m,n=1}^d.$ We refer to such  set $\cS$ as a special set of high weighted GHZ like states on the computational basis.
	The second one is that the  set $\cC(\cS)$ defined by 
	$\{(i_m,j_m,k_m) \mid\sum_{m=1}^dw_d^{(m-1)n}|i_mj_mk_m\rangle \in \cS, n\in\mathbb{Z}_d \}$
	contains the subsets $\{i_0\}\times \mathbb{Z}_{d_2}\times\mathbb{Z}_{d_3},$   $\mathbb{Z}_{d_1}\times \{j_0\}\times\mathbb{Z}_{d_3}$  and  $\mathbb{Z}_{d_1}\times \mathbb{Z}_{d_2}\times\{k_0\}$ for some $i_0\in \mathbb{Z}_{d_1}, j_0\in \mathbb{Z}_{d_2}, k_0\in \mathbb{Z}_{d_3}.$  And we call the set  $\cS$  is plane containing if it satisfies the second condition.

	For each partition $A|BC,B|CA,$ and $C|AB$, we attach the set $\cS$ with a graph $\cG_{A}(\cS)$, $\cG_{B}(\cS)$ and $\cG_{C}(\cS)$ respectively. Here we give the exact description of $\cG_B(\cS)$ as example, $\cG_{A}(\cS)$ and $\cG_{C}(\cS)$ can be defined similarly.  The graph  $\cG_B(\cS)$ is determined by its vertexes $V_B(\cS)$  and  edges $E_B(\cS).$ Here $V_B(\cS)$ is defined to be $\mathbb{Z}_{d_3}\times \mathbb{Z}_{d_1}.$   For each $d$-tuples of GHZ like states with weight $d$, say $
	\big\{\sum_{m=1}^dw_d^{(m-1)n}|i_mj_mk_m\rangle \mid n\in\mathbb{Z}_d\}\subseteq \cS$, it contributes to the edges  $\{(k_m,i_m),(k_n,i_n)\}$ whenever $1\leq m\neq n\leq d$ (which is equivalent to that the subgraph with nodes $\{(k_m,i_m)\}_{m=1}^d$ is a complete graph).
	The edges  $E_B(\cS)$ of $\cG_{A}(\cS)$ are completely determined by the edges obtained in this form.

	\begin{theorem}\label{thm:OneUniform_SN_tri}
		Let $\cS$ be a special orthogonal   set of  high weighted GHZ like states     under computational basis that is plane containing in $\mathbb{C}^{d_1}\otimes\mathbb{C}^{d_2}\otimes \mathbb{C}^{d_3}$. If all  the    corresponding graphs of  $\cS$    are connected,  then the set  $\cS$ is of the strongest nonlocality.
	\end{theorem}	
	
	The proof  of  Theorem \ref{thm:OneUniform_SN_tri}  is  given  in  Appendix. Note   if a subgraph (which contains all vertexes)  of a  graph is connected then the graph itself is also connected. Now we define a subgraph $\hat{\cG}_B(\cS)$ of ${\cG}_B(\cS)$. The nodes of $\hat{\cG}_B(\cS)$ are also $V_B(\cS)$. But the edges of $\hat{\cG}_B(\cS)$ come  from the following way:  For each $d$-tuples of GHZ like states with weight $d$, say $
	\big\{\sum_{m=1}^dw_d^{(m-1)n}|i_mj_mk_m\rangle \mid n\in\mathbb{Z}_d\}\subseteq \cS$ (by reordering, one can assume that $k_m$is incremental), it contributes to the edges  $\{(k_m,i_m),(k_{m+1},i_{m+1})\}$ whenever $1\leq m\leq d-1$. $\hat{\cG}_A(\cS)$ and $\hat{\cG}_C(\cS)$  can be defined similarly. We present an construction of a basis of GHZ like states with weight $4$ in $\mathbb{C}^4\otimes \mathbb{C}^4\otimes \mathbb{C}^4$ whose corresponding graphs are all connected.
	
	\begin{example}\label{GES_SN_444}
	In  $\mathbb{C}^4\otimes \mathbb{C}^4\otimes \mathbb{C}^4$, the set
	$\cS:=\cup_{i=1}^{16}\cS_i$  given by Table \ref{example444-uniform} is  an  \emph{OGES} of  the strongest nonlocality. The size of this set is $64$.
\end{example}	
	\begin{table}[h]	
		$\begin{array}{lcccclcccc}\hline\hline
			\text{Set}& C_1 &C_2 & C_3& C_4 & \text{Set}& C_1 &C_2 & C_3& C_4\\ \hline
			\cB_1&|000\rangle&|121\rangle&|212\rangle&|333\rangle& \cB_9&|013\rangle&|132\rangle&|201\rangle&|320\rangle\\
			\cB_2&|003\rangle&|111\rangle&|222\rangle&|330\rangle &\cB_{10}&|021\rangle&|133\rangle&|200\rangle&|312\rangle\\
			\cB_3&|030\rangle&|112\rangle&|221\rangle&|303\rangle &\cB_{11}&|022\rangle&|101\rangle&|233\rangle&|310\rangle\\
			\cB_4&|033\rangle&|122\rangle&|211\rangle&|300\rangle &\cB_{12}&|023\rangle&|100\rangle&|232\rangle&|311\rangle\\
			\cB_5&|001\rangle&|113\rangle&|230\rangle&|322\rangle &	\cB_{13}&|010\rangle&|131\rangle&|223\rangle&|302\rangle\\
			\cB_6&|002\rangle&|123\rangle&|210\rangle&|331\rangle &\cB_{14}&|020\rangle&|102\rangle&|231\rangle&|313\rangle\\
			\cB_7&|011\rangle&|103\rangle&|220\rangle&|332\rangle&\cB_{15}&|031\rangle&|110\rangle&|202\rangle&|323\rangle\\
			\cB_8&|012\rangle&|130\rangle&|203\rangle&|321\rangle&
			\cB_{16}&|032\rangle&|120\rangle&|213\rangle&|301\rangle\\
			\hline\hline
			
		\end{array}$
		\caption{This table shows the 1-uniform states in $\mathbb{C}^4\otimes\mathbb{C}^4\otimes \mathbb{C}^4$. Each $\cB_i$ contains four coordinately different elements of the computational basis. And it corresponds to four 1-uniform states $\cS_i$ which are linear combinations of the four vectors with coeffients in the rows of $F_4$.}\label{example444-uniform}

	\end{table}	
	\begin{figure}[h]
	\centering
	\includegraphics[scale=0.45]{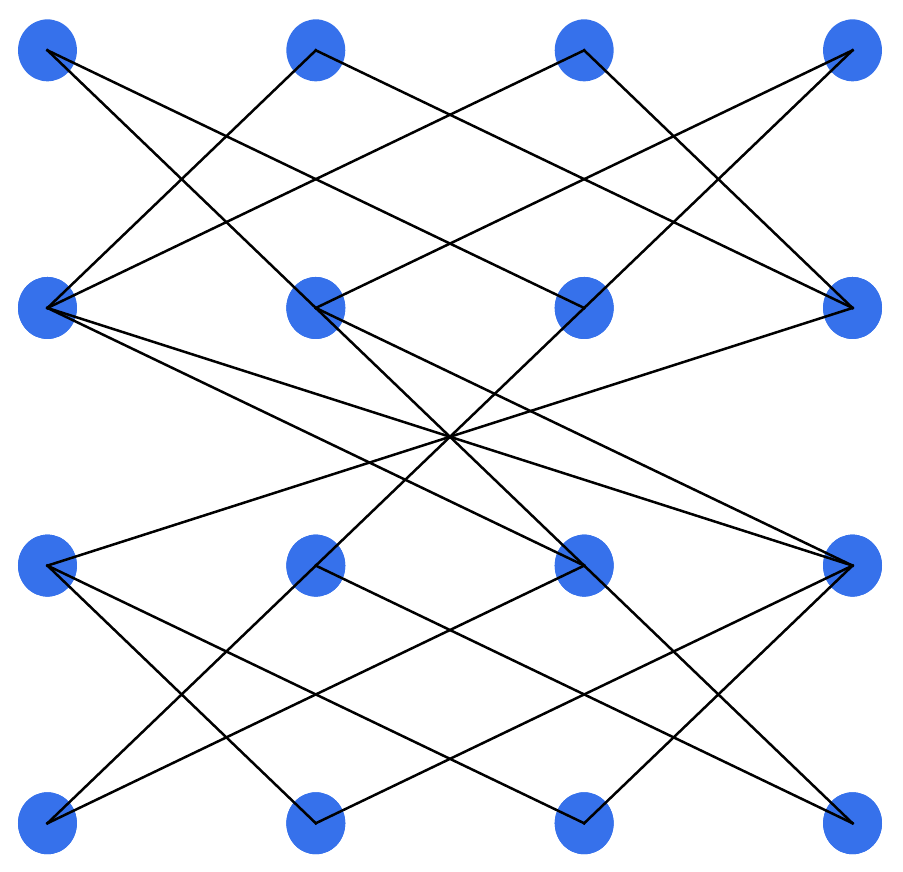}
	\caption{This shows the   graph $\hat{\cG}_A(\cS)$ corresponding  to the OGES  of $\mathbb{C}^4\otimes \mathbb{C}^4\otimes \mathbb{C}^4$ defined by Table \eqref{example444-uniform}. }\label{fig:444uniform}
\end{figure}

	We show the corresponding graph $\hat{\cG}_A(\cS)$ of $\cS$ in the      Fig. \ref{fig:444cube}. Clearly, it is connected.   One can check that the graphs $\hat{\cG}_B(\cS)$and $\hat{\cG}_C(\cS)$ share the same graph as $\hat{\cG}_A(\cS)$. Hence $\cS$ is of the strongest nonlocality.

	\section{Conclusion and Discussion}\label{fifth}
	
	We study the strongest nonlocality of set which is only containing genuinely entangled states.  We show that the   strongest nonlocality of some special set of genuinely entangled states is rather simple. In fact, it  is equivalent to check whether its corresponding graphs are connected or not. On the other hand, we can construct some special set of GHZ like states in $\mathbb{C}^d\otimes \mathbb{C}^d\otimes \mathbb{C}^d$ whose corresponding graphs are all connected. Therefore, we have successfully constructed set of  GHZ like states that is of the strongest nonlocality.  This give answers to a question asked by Shi et al. in reference \cite{Shi20S}.   This  also provides another answer,  which is different from that of  Shi et al's, to  an open question raised by Halder et al. in reference \cite{Halder19}.
	
	There are also some questions left to be considered. For examples, whether can we construct some smaller set with the strongest nonlocality via the OGES than the OPS?  Whether the absolutely entangled states can present the phenomenon of strong nonlocality for four or more parties systems?

	\vspace{2.5ex}
	
	\noindent{\bf Acknowledgments}\, \,   This  work  is supported  by  National  Natural  Science  Foundation  of  China  (11901084, 12005092,   and  U1801661), the China Postdoctoral Science Foundation (2020M681996),  
	the Research startup funds of DGUT (GC300501-103), the  Natural  Science  Foundation  of  Guang-dong  Province  (2017B030308003),   the  Key  R$\&$D  Program of   Guangdong   province   (2018B030326001),   the   Guang-dong    Innovative    and    Entrepreneurial    Research    TeamProgram (2016ZT06D348), the Science, Technology and   Innovation   Commission   of   Shenzhen   Municipality (JCYJ20170412152620376   and   JCYJ20170817105046702 and  KYTDPT20181011104202253),   the Economy, Trade  and  Information  Commission  of  Shenzhen Municipality (201901161512).
	
	\vskip 10pt

	{
		$${\text{\textbf{APPENDIX  }}}$$
	}
	\noindent{\textbf{Proof of Proposition \ref{pro:GES_SN_tri_o}}:}  The structure of the set $\cS$ has similar geometric picture  in Fig.\ref{fig:333cube}. Therefore, $\cS$ is plane containing. By  the symmetricity of the  construction of $\cS$, we only need to show that $\cG_A(\cS)=(V_A(\cS),E_A(\cS))$ is connected.
	
	By definition, $V_A(\cS)=\mathbb{Z}_d\times\mathbb{Z}_d$. From those states belong to $\cS_1$, we know that $\{(i,j+1),(i+1,j)\}\in E_A(\cS)$ for $i,j\in\mathbb{Z}_{d-1}$.
	From those states belong to $\cS_2$, we can obtain that $\{(1,j),(\hat{d},j+1)\}\in E_A(\cS)$ for $j\in\mathbb{Z}_{\hat{d}}$. Using these results, one can obtain that for any nodes $(i,j)$ and $(k,l)$ in $V_A(\cS)$, if $i+j\equiv k+l \ \mathrm{ mod }\  d$, then they must belong to some  same connected component. As a result, there are at most $d$ connected components $C_{[i]}:=\{(j,k)\in V_A(\cS)\mid j+k \equiv i \text{ mod } d\}$ for $i\in\mathbb{Z}_d$. On the other hand, the states in $\cS_3$ implies that $(j,\hat{d})$ and $(j+1,0)$   also belong to some same connected component when $j\in\mathbb{Z}_{\hat{d}}$. Therefore, $C_{[j-1]}$ and $C_{[j+1]}$   are in the same connected component for $j=0,1,\cdots,d-2$.   With these at hand, one can show that these sets are connected in the following way:		{\small\begin{equation*}\label{eq:component}
		C_{[0]}\overset{j=1}{-} C_{[2]}\overset{j=3}{-}\cdots \overset{j=d-2}{-} C_{[d-1]}\overset{j=0}{-}C_{[1]}\overset{j=2}{-}C_{[3]} \overset{j=4}{-}\cdots \overset{j=d-3}{-} C_{[d-2]}.
	\end{equation*}
}   Therefore, there are only one connected component in total and $\cG_A(\cS)$ is connected. \qed

	\vskip 8pt
	
	\noindent{\textbf{Proof of Proposition \ref{pro:GES_SN_tri_e}   }:}  For each $X=A,B,C$ and $i\in \mathbb{Z}_d$, denote  $C^X_{[i]}:=\{(j,k)\in V_X(\cS)\mid j+k \equiv i \text{ mod } d\}.$ Using $\cS_1, \cS_2$ and $\cS_3$, with similar argument as above, one can show that    $C^X_{[j-1]}$ and $C^X_{[j+1]}$   are in the same connected component for  $j=0,1,\cdots,d-2$.  As $d-2$ is even but $d-3$ is odd, we have
	{\small $$\begin{array}{l}
			C^X_{[0]}\overset{j=1}{-} C^X_{[2]}\overset{j=3}{-}\cdots \overset{j=d-3}{-} C^X_{[d-2]},\\[2mm]
			C^X_{[1]}\overset{j=2}{-}C^X_{[3]} \overset{j=4}{-}\cdots \overset{j=d-2}{-} C^X_{[d-1]}\overset{j=0}{-} C^X_{[1]}.
		\end{array}
	$$
	} 
	Hence, there are at most two connected components for each  $\cG_A(\cS),\cG_B(\cS),\cG_C(\cS)$ and  they are determined by the  parity of  $i+j$ for a node $(i,j)\in V$. We complete the proof into three cases as follows.

	\begin{enumerate}[(a)] 
		
		\item $\cG_A(\cS)$.  By those states in $\cS_4$, we know that $(0,0)$ and $(2,3)$ belong to the same connected component. Therefore, the two connected components concoide with each other. 
		\item $\cG_B(\cS)$.  By those states in $\cS_5$, we know that $(\hat{d},\hat{d})$ and $(3,2)$ belong to the same connected component. Therefore, the two connected components concoide with each other.
		
		\item $\cG_C(\cS)$.  By those states in $\cS_4$, we know that $(0,0)$ and $(2,3)$ belong to the same connected component. Therefore, the two connected components concoide with each other.
	\end{enumerate}
	
	\qed

	\vskip 5pt
	
	\noindent{\bf Proof of Theorem \ref{thm:OneUniform_SN_tri}:} 	Without loss of generality,   suppose that  $B$ and $C$ come together to   perform a joint  orthogonality-preserving POVM named $\{E=M^{\dagger}M\}$, where $E=(a_{ij,k\ell})_{i,k\in \mathbb{Z}_{d_2},j,\ell\in \mathbb{Z}_{d_3}}$. The postmeasurement states $\{\mathbb{I}_A\otimes M\ket{\psi}\big |\ket{\psi}\in \cS\}$ should be mutually orthogonal,  i.e., 
	\begin{equation}\label{eq:OGESgenrelationthm2}
		\langle \phi|\mathbb{I}_A\otimes E |\psi\rangle=0
	\end{equation}  
	for $|\psi\rangle, |\phi\rangle \in \cS$ and  $|\psi\rangle\neq |\phi\rangle.$  First, we  show that the matrix $E$ is diagonal.    For any  pair of non-equal  coordinates  $(k,l)$ and $(k',l')$  in $\mathbb{Z}_{d_2}\times \mathbb{Z}_{d_3}$,  as  $\cS$ is  plane containing,   there exist  some $i_0\in \mathbb{Z}_{d_1}$ and two sets  of high weight GHZ like states 
	$$
	\begin{array}{l}
		\{|\psi_n\rangle:=  \sum_{m=1}^dw_d^{(m-1)n}|i_mj_mk_m\rangle \mid n\in\mathbb{Z}_d\}\subseteq \cS,\\
		\{|\phi_{n'}\rangle:=  \sum_{m'=1}^{d'}w_{d'}^{(m'-1)n'}|i'_{m'}j'_{m'}k'_{m'}\rangle \mid n'\in\mathbb{Z}_{d'}\}\subseteq \cS
	\end{array}
	$$ 
	such that $|i_0kl\rangle$ and $|i_0k'l'\rangle$ is one of the summation  terms of $|\psi_n\rangle$ and $|\phi_{n'}\rangle$  respectively (they can not appear in just one  GHZ like state by our definition).     Therefore, we have	\begin{equation}\label{eq:OGESgeneraldia}
		\langle \phi_{n'}|\mathbb{I}_A\otimes E |\psi_n\rangle=0,  \text{ for }n\in \mathbb{Z}_{d}, {n'\in \mathbb{Z}_{d'}. }
	\end{equation} 
	Clearly, $|i_0kl\rangle$  and $|i_0k'l'\rangle$ can be written as  some linear combinations of $\{|\psi_n\rangle\}_{n=1}^{d}$ and $\{|\phi_{n'}\rangle\}_{n'=1}^{d'}$ respectively. Therefore, from Eqs.~\eqref{eq:OGESgeneraldia}, one obtains that  
	$$a_{k'l',kl}= {}_A\langle i_0|{}_B\langle k'|{}_C\langle l'| \mathbb{I}_A\otimes E   |i_0\rangle_A|k\rangle_B|l\rangle_C=0.$$
	Therefore, the matrix $E$ must be a diagonal matrix  under the computational basis.
	
	In the following, we prove that the elements on the diagonal of $E$ are equal.  As the graph $\cG_A(\cS)$ is connected, for any two different nodes  $(j_I,k_I)$ and $(j_L,k_L)$  of $\cG_A(\cS)$, there is a path connecting them, namely	
	\begin{equation}\label{eq:paththm2}
		(j_1,k_1)\overset{e_1}{-} (j_2,k_2)\overset{e_2}{-}\cdots \overset{e_{N-1}}{-} (j_{N},k_{N})\overset{e_N}{-}(j_{N+1},k_{N+1})
	\end{equation}  
	here  $(j_1,k_1)=(j_I,k_I)$ and $(j_{N+1},k_{N+1})=(j_L,k_L)$. By the definition of the edges of $\cG_A(\cS)$,   each edge $e_l=\{(j_l,k_l),(j_{l+1},k_{l+1})\}$ must come from some $d^{(l)}$-tuples of weight $d^{(l)}$ GHZ like states $$\{|\varphi_n^{(l)}\rangle:=  \sum_{m=1}^{d^{(l)}}w_{d^{(l)}}^{(m-1)n}|i^{(l)}_mj^{(l)}_mk^{(l)}_m\rangle \mid n\in\mathbb{Z}_{d^{(l)}}\}\subseteq \cS.$$ 
	More exactly, there exist $i_l, i_{l+1} \in \mathbb{Z}_{d_1}$ such that $|i_lj_lk_l\rangle$ and  $|i_{l+1}j_{l+1}k_{l+1}\rangle$ are two  summation terms in $|\varphi_n^{(l)}\rangle$, i.e., $|i_lj_lk_l\rangle=|i^{(l)}_{m_1}j^{(l)}_{m_1}k^{(l)}_{m_1}\rangle$ and  $|i_{l+1}j_{l+1}k_{l+1}\rangle=|i^{(l)}_{m_2}j^{(l)}_{m_2}k^{(l)}_{m_2}\rangle$.
	Appling Eq.~\eqref{eq:OGESgenrelationthm2}  to $|\varphi_0^{(l)}\rangle$ and $|\varphi_n^{(l)}\rangle$, 	 one can 
	obtain that 
	\begin{equation}\label{eq:OGESgenrelationthm2s}
		\langle \varphi_0^{(l)}|\mathbb{I}_A\otimes E |\varphi_n^{(l)}\rangle=0, \text{ for } n\in \mathbb{Z}_{d^{(l)}}\setminus \{0\}.
	\end{equation}
	As $i_m^{(l)}\neq i_n^{(l)}$ whenever $m\neq n$, from the $(d^{(l)}-1)$ equations  in Eqs.~\eqref{eq:OGESgenrelationthm2s}, one deduce that
	$$
	a_{j_m^{(l)}k_m^{(l)},j_m^{(l)}k_m^{(l)}}=a_{j_n^{(l)}k_n^{(l)},j_n^{(l)}k_n^{(l)}}, \ \ \forall m,n \in \mathbb{Z}_{d^{(l)}}.$$
	Particularly, we have $a_{j_{m_1}^{(l)}k_{m_1}^{(l)},j_{m_1}^{(l)}k_{m_1}^{(l)}}=a_{j_{m_2}^{(l)}k_{m_2}^{(l)},j_{m_2}^{(l)}k_{m_2}^{(l)}}$. Hence, $ a_{j_{l}k_l,j_lk_l}=a_{j_{l+1}k_{l+1},j_{l+1}k_{l+1}}.$
	From the path in Eq. \eqref{eq:paththm2} one could easily deduce that 
	$$ a_{j_{I}k_I,j_Ik_I}=a_{j_{L}k_{L},j_{L}k_{L}}.$$
	
	Therefore,  it can be concluded that the matrix $E$ is indeed proportional to identity operator.  \qed

\end{document}